\documentclass[onecolumn,draftclsnofoot]{IEEEtran}
\usepackage{amsthm}

\usepackage[english]{babel}
\usepackage{amssymb,amsmath,color}  %
\usepackage{graphicx} 
\usepackage{algorithm}
\usepackage[noend]{algpseudocode}
\usepackage{cite}
\usepackage[hidelinks]{hyperref}
\usepackage{bookmark}
\usepackage{enumitem}

\renewcommand{\natural}{{\mathbb{N}}} 

\newcommand{\real}{{\mathbb{R}}}

\newcommand{\map}[3]{#1: #2 \rightarrow #3}

\newcommand{\until}[1]{\{1,\ldots,#1\}}

\newcommand{\EE}{\mathcal{E}} 
\newcommand{\GG}{\mathcal{G}}\newcommand{\LL}{\mathcal{L}}

\newcommand{\subj}{\textnormal{subj. to}}

\newcommand{\argmax}{\mathop{\rm argmax}}

\newcommand{\gen}{\texttt{gen}}
\newcommand{\GEN}{\texttt{GEN}}
\newcommand{\stor}{\texttt{stor}}
\newcommand{\STOR}{\texttt{STOR}}
\newcommand{\cload}{\texttt{conl}}
\newcommand{\CLOAD}{\texttt{CONL}}
\newcommand{\des}{\texttt{des}}
\newcommand{\trade}{\texttt{tr}}

\newcommand{\nbrs}{\mathcal{N}}
\newcommand{\relint}{\mathop{\bf relint}} %

\newcommand{\1}{\mathbf{1}}

\newcommand{\0}{\mathbf{0}}

\renewcommand{\inf}{\operatornamewithlimits{inf\vphantom{p}}}

\renewcommand{\lim}{\operatornamewithlimits{lim\vphantom{p}}}

\makeatletter
\newcommand{\StatexIndent}[1][3]{%
  \setlength\@tempdima{\algorithmicindent}%
  \Statex\hskip\dimexpr#1\@tempdima\relax}
\makeatother

\algnewcommand{\algorithmicgoto}{\textbf{go to }}%
\algnewcommand{\Goto}[1]{\algorithmicgoto Line~\ref{#1}}%
\algnewcommand{\Label}{\State\unskip}

\usepackage[dvipsnames]{xcolor}

\usepackage{soul}
\usepackage[normalem]{ulem}

\newtheorem{theorem}{Theorem}[section]
\newtheorem{proposition}[theorem]{Proposition}
\newtheorem{corollary}[theorem]{Corollary}
 \newtheorem{lemma}[theorem]{Lemma}
\newtheorem{remark}[theorem]{Remark}

\newtheorem{assumption}[theorem]{Assumption}
\newcommand\oprocendsymbol{\hbox{$\square$}}
\newcommand\oprocend{\relax\ifmmode\else\unskip\hfill\fi\oprocendsymbol}

\usepackage{tikz}
\usetikzlibrary{shapes,arrows}

\graphicspath{{figs/}}

\def \algname/{Relaxation and Successive Distributed Decomposition}
\def \algacronym/{RSDD}

\newcommand{\sx}[1]{\mathbf{x}_{#1}}
\newcommand{\bx}{\mathbf{x}}

\newcommand{\bg}{\mathbf{g}}

\newcommand{\bmu}{\boldsymbol{\mu}}

\newcommand{\barsx}[1]{\bar{\mathbf{x}}_{#1}}

\newcommand{\smu}[1]{\boldsymbol{\mu}_{#1}}

\newcommand{\slambda}[1]{\boldsymbol{\lambda}_{#1}}

\newcommand{\tildeslambda}[1]{\boldsymbol{\widetilde{\lambda}}{}_{#1}}
\newcommand{\sLambda}{\boldsymbol{\Lambda}}

\newcommand{\tildesLambda}{\boldsymbol{\widetilde{\Lambda}}}

\newcommand{\domainEta}{D_{\eta}}
\newcommand{\domainEtaRelaxed}{D_{\eta }^{ \texttt{NR} }}

\definecolor{blue@O4S}{RGB}{0, 41, 69}
\definecolor{emph@O4S}{RGB}{0, 93, 137}
\definecolor{red@O4S}{RGB}{127,0,0}
\definecolor{gray@O4S}{RGB}{112, 112, 112}

\usepackage{accents}
\newcommand \ubar[1]{%
  \underaccent{\bar}{#1}}

\begin{document}

\title{
  Constraint Coupled Distributed Optimization:\\
  a Relaxation and Duality Approach
}

\author{Ivano Notarnicola,~\IEEEmembership{Student Member,~IEEE} 
    and Giuseppe Notarstefano,~\IEEEmembership{Member,~IEEE} 
    \thanks{A preliminary short version of this paper 
    has appeared as~\cite{notarnicola2017duality}: 
    the current article includes a much improved comprehensive treatment, all the 
    theoretical proofs and a numerical example on a microgrid control problem.}
  \thanks{Ivano Notarnicola and Giuseppe Notarstefano are with the Department of Engineering,
    Universit\`a del Salento, Via Monteroni, 73100
    Lecce, Italy, \texttt{name.lastname@unisalento.it.} This result is part of a
    project that has received funding from the European Research Council (ERC)
    under the European Union's Horizon 2020 research and innovation programme
    (grant agreement No 638992 - OPT4SMART). } }

\maketitle
\begin{abstract}
In this paper we consider a general, challenging distributed optimization set-up arising in several important network control applications. Agents of a network want to minimize the sum of local cost functions, each one depending on a local variable, subject to local and coupling constraints, with the latter involving all the decision variables. We propose a novel fully distributed algorithm based on a relaxation of the primal problem and an elegant exploration of duality theory. Despite its complex derivation, based on several duality steps, the distributed algorithm has a very simple and intuitive structure. That is, each node finds a primal-dual optimal solution pair of a local, relaxed version of the original problem, and then updates suitable auxiliary local variables. We prove that agents asymptotically compute their portion of an optimal (feasible) solution of the original problem. This primal recovery property is obtained without any averaging mechanism typically used in dual decomposition methods. To corroborate the theoretical results, we show how the methodology applies to an instance of a Distributed Model Predictive Control scheme in a microgrid control scenario.
\end{abstract}

\section{Introduction}
\label{sec:intro}
In the last decade distributed optimization has received significant
attention. Literature has mainly focused on \emph{cost-coupled} optimization
problems in which the cost to be minimized is the sum of local functions depending
on a common decision variable, 
see~\cite{nedic2009distributed,nedic2010constrained,%
  duchi2012dual,zhu2012distributed,jakovetic2014fast,nedic2015distributed,dilorenzo2016next,
  varagnolo2016newton,notarnicola2017asynchronous}
and references therein for an overview.
A different, more general optimization set-up amenable to distributed
computation is the minimization of the sum of local cost functions, each one
depending on a \emph{local} variable, subject to a local constraint for each
variable and a coupling constraint involving all the decision variables.
In this problem, the global optimal solution is obtained 
by stacking all the local variables. This feature leads easily to so-called 
big-data problems having a very highly dimensional decision variable that
grows with the network size. 
However, since agents are typically interested in computing only their (small) 
portion of the optimal solution, novel tailored methods 
need to be developed to address these challenges.
We call this framework \emph{constraint-coupled} optimization set-up.

Several scenarios of interest in Controls and Robotics as well as Communication and Signal
Processing strongly motivate the investigation of such a problem.
Example set-ups include resource allocation (e.g., in Communication
or Cooperative Robotics) or network flow optimization (e.g., in smart grid
energy management).
A set-up which is particularly relevant in our community is distributed Model
Predictive Control (MPC) in which the goal is to design a feedback law for a
(spatially distributed) network of dynamical systems based on the MPC
concept. In such a scheme several optimization problems need to be iteratively
solved. The local decision variable of each agent corresponds to its
state-input trajectory, while the local constraints encode its dynamics, which
is usually independent of other agents. A constraint that couples agents'
states, inputs or outputs needs to be taken into account in order to enforce
cooperative tasks as, e.g., formation control, or to take into account common
bounds, e.g., due to shared resources. 
Distributed MPC approaches are mainly
classified into non-cooperative and cooperative schemes, see, e.g.,
\cite{christofides2013distributed}. While in non-cooperative
schemes the main focus is to guarantee recursive feasibility and stability,
in cooperative approaches agents care also about optimality 
when solving the global, constraint-coupled
problems arising in each time window
and, thus, call for tailored distributed optimization algorithms.

Parallel methods for constraint-coupled problems have been developed mainly in
the context of cooperative MPC.
They are based on a master-subproblem architecture that 
traces back to late 90s \cite{bertsekas1989parallel}.
Duality is a widely used tool to decompose the problem and design
optimization algorithms as shown, e.g., in the
tutorial papers~\cite{palomar2006tutorial,necoara2011parallel}.
In~\cite{giselsson2013accelerated} an accelerated dual
decomposition is proposed to solve a MPC problem.
In~\cite{dinh2013dual} dual decomposition is combined with a penalty 
approach to solve separable nonconvex problems.
A linear convergence rate for a dual gradient algorithm for
linearly constrained separable convex problems is proven  
in~\cite{necoara2015linear}.
In~\cite{tran2016fast} an inexact dual decomposition scheme combined with
smoothing techniques is proposed.
In~\cite{hours2016parametric} a primal-dual, real-time strategy is proposed 
to solve parametric nonconvex programs usually arising in nonlinear MPC.
Recently, parallel augmented Lagrangian methods have been proposed
in~\cite{houska2016augmented,chatzipanagiotis2017convergence} to 
solve nonconvex problem instances with linear coupling constraints.
Although sometimes termed as ``distributed'' algorithms,
they require a centralized unit performing some steps in the 
proposed strategies.
When further sparsity is assumed in the problem, e.g., 
the sub-systems have coupled dynamics with their neighbors only
\cite{necoara2011parallel,dinh2013dual,necoara2015linear,giselsson2013accelerated},
then the parallel scheme can be implemented over a network
without a central authority. In this paper we propose a purely distributed 
algorithm also for general coupling constraints that involve the entire set of agents.

Distributed optimization algorithms for special versions of the
constraint-coupled set-up, arising in the context of resource allocation
problems, have been proposed in
\cite{lakshmanan2008decentralized,simonetto2012regularized,necoara2013random,
  cherukuri2015distributed}.
The general constraint-coupled set-up we consider in this paper has not received
extensive investigation in a purely distributed framework and only few works are
available, 
i.e.,~\cite{chang2014distributed,simonetto2016primal,falsone2017dual,
mateos2017distributed,notarnicola2016dsm}.
In~\cite{chang2014distributed} a consensus-based primal-dual
perturbation algorithm is proposed to solve smooth constraint-coupled optimization
problems.
Very recently, in~\cite{simonetto2016primal} and \cite{falsone2017dual} distributed algorithms
are proposed based on a consensus-based dual decomposition and a dual proximal
optimization approach, respectively. 
A class of min-max optimization problems, strictly related to the constraint-coupled
set-up, is addressed in~\cite{mateos2017distributed} using a Laplacian-based 
saddle-point subgradient scheme.
A well-known drawback in methods based on dual decomposition, is that primal feasibility is not 
easily retrieved from dual solutions. 
Thus, primal recovery mechanisms are devised in the papers above in order to recover 
a primal solution by suitably applying running 
average schemes borrowed from the centralized literature, see, e.g., \cite{nedic2009approximate}.
A special coupling associated to peak minimization problems arising in 
demand-side management is considered in~\cite{notarnicola2016dsm},
where a (simplified) algorithmic approach, similar to the one proposed in this 
paper, is proposed for that set-up.
An alternative approach to constraint-coupled problems has been proposed
in~\cite{burger2014polyhedral}, and customized to MPC
in~\cite{burger2013non,lorenzen2013distributed}, where agents employ a
cutting-plane consensus scheme to iteratively approximate their local problems.
This approach enforces agents to eventually agree on the complete solution
vector and this may be an undesirable feature in some applications.

The main paper contributions are as follows. We propose a novel, distributed 
optimization algorithm to solve constraint-coupled optimization problems
over networks.
Overall, our distributed algorithm enjoys three appealing features: (i)
local computations at each node involve only the local decision variable and,
thus, scale nicely with respect to the dimension of the decision vector, (ii)
privacy is preserved since agents do not communicate, and thus disclose, their estimates of 
local decision variable, cost or constraints, and
(iii) an estimate of a primal optimal solution component is directly computed by 
each agent without any averaging mechanism, which results in a faster algorithm.

The proposed approach combines a proper relaxation of the 
original problem with an elegant exploration of duality theory.
The resulting distributed algorithm is a two-step procedure in which each agent
iteratively performs a (primal) constrained and small-sized optimization,
followed by a dual update.
The local problems involve the local cost function and the local constraint of
the agent. Also, a local inequality constraint, which is adjusted at each step,
accounts for the coupling constraint involving all the agents.
Although this constraint dynamically changes over the iterations, we do not need
to assume a priori feasibility of local problems, but rather, thanks to the
proposed relaxation approach, local violations are allowed and proven to be
asymptotically vanishing.
Each local solution estimate is guaranteed to asymptotically converge to 
the component of an optimal (and, thus, feasible) solution of the original problem.
Such primal convergence of local estimates, known in the literature 
as primal recovery, is non-obvious in duality-based methods applied to 
(merely) convex programs.
In our distributed algorithm, this property results from the methodology we employed,
without resorting to any (commonly used) running averaging mechanism.
Moreover, this key feature has an even stronger impact on our scheme since 
it allows primal quantities to directly inherit the convergence
rate of a ``centralized'' subgradient iteration, while, in general, running averages further 
degrade such a rate.
Finally, since no particular initialization is required, our distributed algorithm 
can be implemented in a dynamic context in which the problem may change or 
nodes can appear or disappear.

The paper unfolds as follows. In Section~\ref{sec:setup} we formalize
the set-up and introduce our distributed optimization algorithm. 
In Section~\ref{sec:relaxed duality tour} we give a constructive derivation of 
of the algorithm and in~Section~\ref{sec:analysis} we conclude the analysis
by proving the its convergence properties. 
In Section~\ref{sec:simulations} we corroborate the theoretical results by
showing how the methodology applies to an 
instance of a distributed MPC controller for a microgrid.

\section{Distributed Set-up and Optimization Algorithm}
\label{sec:setup}

In this section we formally state the general constraint-coupled problem
we aim at investigating in this paper as strongly motivated by control applications 
discussed in the introduction.
Moreover, we introduce the proposed distributed algorithm along with its convergence theorem.

\subsection{Constraint-Coupled Set-up}

Consider the following optimization problem
\begin{align}
\begin{split}
  \min_{\sx{1},\ldots,\sx{N}} \: & \: \sum_{i=1}^N f_i ( \sx{i} )
  \\
  \subj \: & \: \sx{i} \in X_i, \hspace{1.2cm} i \in\until{N}
  \\
  & \: \sum_{i=1}^N \bg_i (\sx{i}) \preceq \0,
\end{split}
\label{eq:primal_original}
\end{align}
where for all $i\in\until{N}$, the set $X_i\subseteq \real^{n_i}$ with $n_i\in\natural$,
the functions $\map{f_i}{\real^{n_i}}{\real}$ and
$\map{\bg_i}{\real^{n_i}}{\real^S}$ with $S\in\natural$.
The symbol $\preceq$ (and, consistently, for other sides)
means that the inequality holds
component-wise and $\0  \triangleq [0,\ldots,0]^\top \! \! \in \! \real^S\!$.

\begin{assumption}
  For all $i\in\until{N}$, each function $f_i$ is convex, 
  and each $X_i$ is a non-empty, compact, convex set. 
  Moreover, each $\bg_i$ is a component-wise convex function, i.e., 
  for all $s\in\until{S}$ each component $\map{ \bg_{is} }{\real^{n_i}}{\real}$ 
  of $\bg_i$ is a convex function.
  \oprocend
\label{ass:primal_regualirity}
\end{assumption}

The following assumption is the well-known Slater's constraint qualification.
\begin{assumption}
  There exist 
  $\barsx{1}\in\relint(X_1), \ldots, \barsx{N}\in\relint(X_N)$ such that
  $\sum_{i=1}^N \bg_i (\barsx{i}) \prec \0$.~\oprocend
\label{ass:constraints_qualification}
\end{assumption}
These assumptions are quite standard and guarantee 
that problem~\eqref{eq:primal_original} admits (at least) an 
optimal solution $( \bx_1^\star,\ldots, \bx_N^\star)$ such 
that its optimal cost is $\sum_{i=1}^N f_i(\bx_i^\star) = f^\star$.
Moreover, a dual approach can be applied since strong duality 
holds. 
Notice that we assumed that 
$\sum_{i=1}^N \bg_i(\sx{i}) \preceq \0$ admits a strictly feasible point, 
while each $\bg_i(\sx{i}) \preceq \0$ may not.

We consider a network of $N$ processors communicating according to a
\emph{connected} and \emph{undirected} graph $\GG = (\until{N}, \EE)$, where
$\EE\subseteq \until{N} \times \until{N}$ is the set of edges. Edge $(i,j)$
models the fact that node $i$ sends information to $j$. Note that, %
being the graph undirected, for each $(i,j)\in\EE$, then also
$(j,i)\in\EE$. We denote by $|\EE|$ the cardinality of $\EE$ and by $\nbrs_i$
the set of \emph{neighbors} of node $i$ in $\GG$, i.e.,
$\nbrs_i = \left\{j \in \until{N} \mid (i,j) \in \EE \right\}$.
Each node $i$ knows only $f_i$, $\bg_i$ and $X_i$, and aims at estimating its
portion $\sx{i}^\star$ of an optimal solution $( \bx_1^\star,\ldots, \bx_N^\star)$
of~\eqref{eq:primal_original} by means of local communication only. 

\begin{remark}%
The popular set-up $\min_{\sx{} \in \cap_i X_i} \sum_i f_i ( \sx{} )$
can be cast as~\eqref{eq:primal_original}.
However, problem~\eqref{eq:primal_original} is well suited 
to more general frameworks. %
Indeed, in~\eqref{eq:primal_original} 
local variables $\sx{i}$ represent the relevant information agent $i$ is 
interested in, including the special case of being a mere copy of the common $\sx{}$.%
\oprocend
\end{remark}

\subsection{\algname/}
\label{subsec:algorithm}
In this subsection we present our \algname/ method (\algacronym/)
which is a novel strategy to solve problem~\eqref{eq:primal_original}
over networks.

Informally, the algorithm consists of an iterative two-step procedure.  First,
each node $i\in\until{N}$ stores a set of variables $((\sx{i}$, $\rho_{i}), 
\smu{i}) \in \real^{n_i} \times \real \times \real^{S}$
obtained as a primal-dual optimal solution pair of problem
\eqref{eq:alg_minimization}. The vector $\smu{i}$ is the multiplier associated to 
the inequality constraint in~\eqref{eq:alg_minimization}. 
We notice that \eqref{eq:alg_minimization} mimics a local version of the centralized 
problem~\eqref{eq:primal_original}, where the coupling with the other nodes in
the original formulation is replaced by a local term depending only on
neighboring variables $\slambda{ij} \in \real^S$ and $\slambda{ji} \in \real^S$,
$j\in\nbrs_i$. Moreover, this local version of the coupling constraint is also
relaxed, i.e., a positive violation $\rho_{i}\1$ is allowed, where $\1 \triangleq [1,\ldots,1]^\top \in \real^S$. 
Finally, instead of minimizing only the local $f_i$, the (scaled) violation 
$M\rho_i$, $M>0$, enters the cost function as well.
The auxiliary variables
$\slambda{ij}$, $j\in\nbrs_i$, are updated in the second step according to a
suitable linear law that combines neighboring $\smu{i}$ as
in~\eqref{eq:alg_update}.
Nodes use a step-size denoted by $\gamma^t$ and can initialize the
variables $\slambda{ij}$, $j\in\nbrs_i$ to \emph{arbitrary} values.
In the following table we formally state our distributed algorithm from the
perspective of node $i$.
\begin{algorithm}
\renewcommand{\thealgorithm}{}
\floatname{algorithm}{Distributed Algorithm}

  \begin{algorithmic}[0]
    \Statex \textbf{Processor states}: $\sx{i}$, $\rho_{i}$, $\smu{i}$ and $\slambda{ij}$ for $j\in\nbrs_i$

    \Statex \textbf{Evolution}:

      \StatexIndent[0.5] \textbf{Gather} $ \slambda{ji}^t$ from $j\in\nbrs_i$

      \StatexIndent[0.5] \textbf{Compute} $\big( (\sx{i}^{t+1}, \rho_{i}^{t+1}),\smu{i}^{t+1} \big)$ 
        as a primal-dual optimal solution pair of
      \begin{align}
      \begin{split}
        \min_{\sx{i},\rho_{i} } \: & \: f_i ( \sx{i} ) + M  \rho_{i}
        \\
        \subj \: & \: \rho_i \ge 0,\: \: \sx{i} \in X_i
        \\
        & \:  \bg_i (\sx{i}) + \sum_{j\in\nbrs_i} \big( \slambda{ij}^t - \slambda{ji}^t \big) \preceq \rho_{i}\1
      \end{split}
      \label{eq:alg_minimization}
      \end{align}

      \StatexIndent[0.5] \textbf{Gather} $ \bmu_j^{t+1}$ from $j\in\nbrs_i$

      \StatexIndent[0.5] \textbf{Update} for all $j\in\nbrs_i$
      \begin{align}
        \slambda{ij}^{t+1} & = \slambda{ij}^t - \gamma^t \big( \smu{i}^{t+1} - \smu{j}^{t+1} \big)
      \label{eq:alg_update}
      \end{align}

  \end{algorithmic}
  \caption{\algacronym/}
  \label{alg:algorithm}
\end{algorithm}

As already mentioned, each agent $i$ aims at
computing an optimal strategy by means of local interaction only.
The proposed distributed algorithm is a protocol in which agents
exchange only the vectors $\smu{i}^t$ and $\slambda{ij}^t$ without
explicitly communicating the current estimates of their local decision variables
$\sx{i}^t$, costs $f_i$ or constraints $\bg_i$.
This is an important, appealing feature of the \algacronym/ distributed algorithm
since it guarantees privacy in the network.  

\begin{remark}[Arbitrary Initialization]
Since the initialization is arbitrary, then the algorithm can continuously
run into a ``dynamic context'' in which agents can join/leave the network and problem data
can change. These events simply induce a new optimization problem and trigger a new transient
for the distributed algorithm that will eventually converge 
to a solution of the new problem instance.~\oprocend
\end{remark}

In order to gain more intuition about the algorithmic evolution, at this point
we provide an informal interpretation, supported by
Figure~\ref{fig:local_coupling_violation}, of the local optimization step
in~\eqref{eq:alg_minimization}.
\begin{figure}[htpb]
\centering
\includegraphics[scale=0.9]{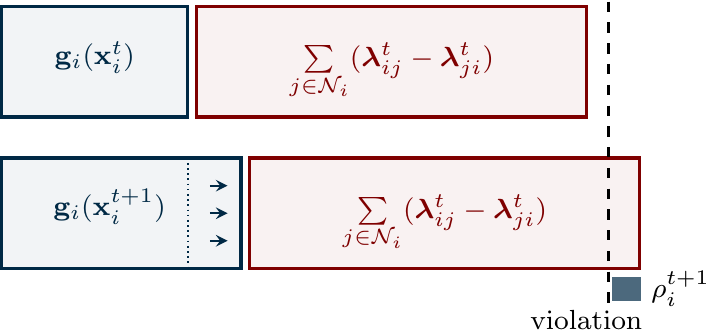}
\caption{
  Graphical representation of the local constraint relaxation for a scalar coupling 
  constraint.
}
\label{fig:local_coupling_violation}
\end{figure}

Agent $i$, due to its partial knowledge, can only optimize with respect to its
own decision variable $\sx{i}$. Thus, it can solve an instance of
problem~\eqref{eq:primal_original} in which all the other variables in the
network have a given value $\bx_{j}^t$ for $j\in\until{N}\setminus \{i\}$. Thus,
the cost function reduces to $f_i$ only.  As for the coupling constraint, it
describes how the resources are allocated to all the agents at each iteration
$t$.  In the figure, we show a possible instance of a feasible allocation: in
blue we depicted the resource assigned to agent $i$ while in red an estimate
of the resources currently allocated to all the other agents. When agent $i$ optimizes 
its local variable $\sx{i}$ only, it can ``play'' with the ``available resource slot'' given by
$-\sum_{j\neq i} \bg_j( \bx_{j}^t )$.
Since the current allocation is in general not optimal, this constraint might be
too restrictive. In fact, it can slow down (and even stall) the optimization
process by easily leading to infeasibility of the local
problem~\eqref{eq:alg_minimization} when $\rho_i$ is set to $0$. On this regard,
recall that we do not assume feasibility of every $\bg_{i}$ independently.
Also, it is worth noting that such ``available resource slot'' depends on the 
entire network's variables, and, thus, it is not an easily
available information in a distributed scenario.
Thus, we propose a strategy in which at each iteration $t$, each agent
$i$ replaces the term $\sum_{j\neq i} \bg_j( \bx_{j}^t )$ in the coupling 
with $\sum_{j\in\nbrs_i} \big( \slambda{ij}^t - \slambda{ji}^t \big)$.
Notice that this term can be computed \emph{locally} at each node by communicating 
with neighboring agents only. 
This term is then iteratively adjusted along the algorithmic evolution 
in order to eventually obtain an optimal solution.
Finally, each agent $i$ is allowed for a violation $\rho_{i}\1$ of the local version of the
coupling constraint. At the same time, this violation is penalized in order to encourage it 
to eventually converge to zero. 
This intuitive description will be rigorously derived and proven in the
following sections.

We are now ready to state the main result of the paper, namely the convergence
of the \algacronym/ distributed algorithm.
We start by formalizing the assumption that the step-size should satisfy.
\begin{assumption}
  The sequence $\{ \gamma^t \}_{t\ge0}$, with each $\gamma^t \ge 0$, satisfies the
  conditions $\sum_{t=0}^{\infty} \gamma^t = \infty$,
  $\sum_{t=0}^{\infty} \big( \gamma^t \big)^2 < \infty$.
  \label{ass:step-size}
  \oprocend
\end{assumption}

The convergence theorem is stated below.
\begin{theorem}
  Let Assumption~\ref{ass:primal_regualirity} and
  \ref{ass:constraints_qualification} hold, and let the step-size $\gamma^t$
  satisfy Assumption~\ref{ass:step-size}.
  Moreover, letting $\smu{}^\star$ be an optimal solution of the dual
  of problem~\eqref{eq:primal_original}, assume
  $M > \| \smu{}^\star\|_1$.
  Consider a sequence $\big\{ \sx{i}^t, \rho_{i}^t \big\}_{t\ge 0}$, $i\in \until{N}$, generated by the 
  \algacronym/ distributed algorithm with an arbitrary initialization.
  Then, the following holds:
  \begin{enumerate}
    \item $\big \{ \sum_{i=1}^N \big( f_i ( \sx{i}^t ) + M \rho_{i}^t \big) \big\}_{t\ge 0}$ 
    converges to the optimal cost $f^\star$ of~\eqref{eq:primal_original};

    \item every limit point of $\big \{ \sx{i}^t \big\}_{t\ge 0}$,
    $ i \in \until{N}$, 
    is a primal optimal (feasible) solution of~\eqref{eq:primal_original}.%
  \end{enumerate}
\label{thm:convergence}  
\end{theorem}
\begin{proof}
The proof is given in Section~\ref{subsec:convergence_proof}.
\end{proof}

\begin{remark}
  When $f_i = 0$ for all $i\in\until{N}$, then our \algacronym/ becomes a
  distributed algorithm for solving a feasibility problem, i.e., find
  $(\sx{1},\ldots,\sx{N})$ such that $\sx{i} \in X_i$, for all $i \in\until{N}$
  and $\sum_{i=1}^N \bg_i (\sx{i}) \preceq \0$.  \oprocend
\end{remark}

\section{Algorithm Analysis: \\ Relaxation and Duality Tour}
\label{sec:relaxed duality tour}
In this section we present a constructive derivation of our distributed algorithm.
The methodology relies on a proper relaxation of the original problem
and on the derivation of a sequence of equivalent problems.
We point out that, although based on a relaxation, the
proposed algorithm exploits such a relaxation to solve \emph{exactly} the 
original problem formulation in a distributed way.

\subsection{First Duality Step and Relaxation Approach}
\label{subsec:first_dual}
We start our duality tour by deriving the dual problem
of~\eqref{eq:primal_original} and a restricted version necessary for the
algorithm derivation.

Let $\smu{} \succeq \0 \in\real^S$, be a multiplier associated to 
the inequality constraint $\sum_{i=1}^N \bg_{i} (\sx{i} ) \preceq \0$ in~\eqref{eq:primal_original}.
Then, the dual of problem~\eqref{eq:primal_original} is given by
\begin{align}
  \begin{split}
    \max_{\smu{} \in \real^S} \: & \: \sum_{i=1}^N q_i (\smu{})
    \\
    \subj \: & \:\: \smu{} \succeq \0
  \end{split}
\label{eq:dual_original}
\end{align}
where each term $q_i$ of the dual function $q(\smu{}) = \sum_{i=1}^N q_i (\smu{})$ is 
defined as
\begin{align}
    q_i (\smu{}) = \min_{ \sx{i} \in X_i} \Big( f_i ( \sx{i} ) + \smu{}^\top \bg_i( \sx{i} ) \Big),
    \label{eq:qi_definition}
\end{align}
for all $i\in\until{N}$. Let $q^\star$ be the optimal cost of~\eqref{eq:dual_original}.

As already mentioned, in light of Assumptions~\ref{ass:primal_regualirity} 
and~\ref{ass:constraints_qualification}, problem~\eqref{eq:primal_original} is
feasible and has finite optimal cost $f^\star$. 
Moreover, the Slater's condition holds and, thus, the strong duality theorem for convex 
inequality constraints, \cite[Proposition~5.3.1]{bertsekas1999nonlinear}, applies, ensuring 
that strong duality holds, i.e., problems~\eqref{eq:primal_original} and~\eqref{eq:dual_original} 
have the same optimal cost $f^\star = q^\star$.
Moreover, $q^\star$ is attained at some $\smu{}^\star \succeq \0$, i.e., $q( \smu{}^\star ) = q^\star$,
cf.~\cite[Proposition~5.1.4]{bertsekas1999nonlinear}.
Finally, we recall that since $\sum_{i=1}^N q_i (\smu{})$ is the dual function
of~\eqref{eq:primal_original}, then it is concave on its convex domain $\smu{} \succeq \0$. 
With the dual problem at hand, several existing algorithms can be applied to directly
solve~\eqref{eq:dual_original} in a distributed way, see e.g., the distributed
projected subgradient~\cite{nedic2009distributed}. However, as pointed 
out in the introduction such dual approaches do not guarantee primal recovery and 
additional schemes must be employed to regain it, e.g., averaging mechanisms.

In this paper we propose an alternative approach that relies on a further duality 
step that gives rise to the \algacronym/ distributed algorithm, which overtakes 
these issues.
Let us introduce an optimization problem similar to~\eqref{eq:dual_original},
given by
\begin{align}
  \begin{split}
    \max_{\smu{} \in \real^S} \: & \: \sum_{i=1}^N q_i (\smu{})
    \\
    \subj \: & \: \smu{} \succeq \0, \: \smu{}^\top \1 \le M,
  \end{split}
  \label{eq:dual_restricted}
\end{align}
where $M$ is a positive scalar and $\1=[1,\ldots,1]^\top$. %
This problem is a \emph{restricted} version of problem~\eqref{eq:dual_original}. 
Here, in fact an additional constraint, namely $\smu{}^\top \1 \le M$, has been 
added to~\eqref{eq:dual_original}. It is worth mentioning 
that this restriction makes the constraint set of~\eqref{eq:dual_restricted} 
a compact set.
Although this step may seem not motivated at this point of the paper, its necessity
will be clear from the following steps of the analysis, see also 
Section~\ref{subsec:on_the_relaxation} for a dedicated discussion.

Notice that, if $M$ is sufficiently large, the presence of the constraint
$\smu{}^\top \1 \le M$ in~\eqref{eq:dual_restricted} will not alter the optimal 
solutions of the unrestricted problem~\eqref{eq:dual_original}. 
The next result formally establishes the relationship between
problems~\eqref{eq:dual_restricted} and~\eqref{eq:dual_original}.
\begin{lemma}
\label{lem:relaxation_equivalence}
  Let $\smu{}^\star$ be an optimal solution of problem~\eqref{eq:dual_original} and
  $M$ be a positive scalar satisfying $M > \| \smu{}^\star \|_1$.
  Then, problems~\eqref{eq:dual_restricted} and~\eqref{eq:dual_original}
  have the same optimal cost, namely $q^\star = f^\star$.
  Moreover, $\smu{}^\star$ is an optimal solution also for problem~\eqref{eq:dual_restricted}.
\end{lemma}
\begin{proof}
  The constraint set $\{ \smu{} \succeq \0 \mid \smu{}^\top \1 \le M\}$ is a restriction
  of the constraint set $\smu{} \succeq \0$ of problem~\eqref{eq:dual_original} containing $\smu{}^\star$. 
  Thus the optimal cost of~\eqref{eq:dual_restricted} is, in general, greater than or equal to the
  optimal cost of~\eqref{eq:dual_original}. Since the domain of~\eqref{eq:dual_restricted} contains at 
  least one optimal solution of problem~\eqref{eq:dual_original}, namely $\smu{}^\star$, then the 
  optimal cost of problem~\eqref{eq:dual_restricted} is $q^\star$ and is (at
  least) attained at $\smu{}^\star$, so that the proof follows.%
\end{proof}

With the dual problem~\eqref{eq:dual_original} and its restricted
version~\eqref{eq:dual_restricted} at hand, one may wonder about the connection
between their primal counterparts.
Next, we show that the restricted problem~\eqref{eq:dual_restricted} is the dual 
of a \emph{relaxed} version of the original primal problem~\eqref{eq:primal_original}. 
\begin{lemma}
\label{lem:primal_relaxation}
  Problem~\eqref{eq:dual_restricted} is the dual of the following optimization 
  problem
  \begin{align}
  \begin{split}
    \min_{ \sx{1},\ldots,\sx{N}, \rho } \: & \: \sum_{i=1}^N f_i ( \sx{i} ) + M \rho
    \\
    \subj \: & \: \rho \ge 0, \:\: \sx{i} \in X_i, \:\:\: i \in\until{N}
    \\
    & \: \sum_{i=1}^N \bg_i (\sx{i}) \preceq \rho \1,
  \end{split}
  \label{eq:primal_relaxed}
  \end{align}
  and strong duality holds.%
\end{lemma}
\begin{proof}
  The dual function of~\eqref{eq:primal_relaxed} is given by
  \begin{align*}
    q_R(\smu{}) \! &=\!\!
      \inf_{ \substack{ \sx{1}\in X_1,\ldots,\sx{N}\in X_N\\ \rho \ge 0} }   
      \!  \sum_{i=1}^N \!\!\Big( \! f_i (\sx{i}) \!+\! \smu{}^{\!\top \! \!} \bg_i (\sx{i}) \!\Big) 
    \!\!+\!\!
    \rho \big( M \!-\! \smu{}^{\!\top } \1 \big)
    \\[0.5ex]
    & =
    \begin{cases}
      \displaystyle \sum_{i=1}^N \underbrace{ \min_{ \sx{i}\in X_i } \!\! \Big( f_i (\sx{i}) + \smu{}^{\!\top } \bg_i
      (\sx{i}) \Big) }_{ q_i(\smu{}) }, \hspace{-0.2cm} & \text{if} \, M \!-\! \smu{}^\top \1 \ge 0
      \\[-0.5ex]
      -\infty, & \text{otherwise}
    \end{cases}
  \end{align*}
  where each $q_i(\smu{}) $ is the same defined in~\eqref{eq:qi_definition}.
  The maximization of the dual function $q_R(\smu{})$ on its domain 
  turns out to be the maximization of $\sum_{i=1}^N q_i (\smu{})$ over
  $\{ \smu{} \succeq \0 \mid \smu{}^\top \1 \le M\}$, which is
  problem~\eqref{eq:dual_restricted}, and the proof follows.%
\end{proof}

Notice that problem~\eqref{eq:primal_relaxed}
is a relaxation of problem~\eqref{eq:primal_original} since we allow for 
a positive violation of the coupling constraint. At the same time, the %
violation $\rho$ is penalized with a scaling factor $M$ in order to discourage it.
The variable $\rho$ resembles the $\rho_{i}$ introduced in the distributed
\algacronym/ algorithm. However, as it will be clear from the forthcoming analysis,
$\rho_{i}$ is \emph{not} a local estimate (or copy) of $\rho$, but it rather represents
the local contribution of agent $i$ to the common violation $\rho$.

The following result characterizes how the original primal problem~\eqref{eq:primal_original} and 
its relaxed version~\eqref{eq:primal_relaxed} are related.
\begin{proposition}
\label{prop:relaxed_solution_structure}
  Let $M$ be such that $M > \| \smu{}^\star\|_1$ with
  $\smu{}^\star$ an optimal solution of the dual of problem~\eqref{eq:primal_original}.
  The optimal solutions of the relaxed problem~\eqref{eq:primal_relaxed} 
  are in the form $(\sx{1}^\star,\ldots,\sx{N}^\star,0)$, where $(\sx{1}^\star,\ldots,\sx{N}^\star)$
  is an optimal solution of~\eqref{eq:primal_original}, i.e., solutions of~\eqref{eq:primal_relaxed}
  must have $\rho^\star = 0$.
\end{proposition}
\begin{proof}
  First, we notice that problem~\eqref{eq:primal_relaxed}
  is the epigraph formulation of
  \begin{align}
  \begin{split}
    \min_{ \sx{1},\ldots,\sx{N}}  &\! \sum_{i=1}^N f_i ( \sx{i} )
    \!+\! M \max \! \Big\{  0,  \sum_{i=1}^N \bg_{i1} (\sx{i}),\ldots, \sum_{i=1}^N \bg_{iS} (\sx{i}) \! \Big \}
    \\[1.5ex]
    \subj \: & \: \sx{i} \in X_i, \: i \in\until{N},
  \end{split}
  \label{eq:epigraph_problem}
  \end{align}
  where $\bg_{is}$ denotes the $s$-th component of $\bg_{i}$.
  Problems~\eqref{eq:primal_original} and~\eqref{eq:epigraph_problem} enjoy the
  same structure as the ones considered in \cite[Proposition~5.25]{bertsekas1982constrained}.
  By Assumption~\ref{ass:constraints_qualification},
  problem~\eqref{eq:primal_relaxed} (and thus \eqref{eq:epigraph_problem})
  satisfies the assumptions for strong duality, thus we can invoke
  \cite[Proposition~5.25]{bertsekas1982constrained},
  with $c=M$, to conclude that
  problems~\eqref{eq:primal_original} and~\eqref{eq:epigraph_problem} have the
  same optimal solutions, thus completing the proof.%
\end{proof}

\begin{remark}[Alternative restrictions]
  Other choices for the restriction of the domain $\smu{} \succeq \0$ of~\eqref{eq:dual_original} 
  can be considered. 
  For instance, one can consider upper bounds in the form $\smu{} \preceq M\1$ or
  $\smu{} \preceq [M_1, \ldots, M_S]^\top$.
  As one might expect, the specific constraint restriction gives rise to different
  forms of the relaxed primal problem~\eqref{eq:primal_relaxed}.~\oprocend
\end{remark}

\subsection{Second Dual Problem Derivation}
\label{subsec:second_dual}
At this point, we continue our duality tour in order to design an algorithm
that solves problem~\eqref{eq:dual_restricted} instead of the unrestricted
dual problem~\eqref{eq:dual_original}.

In order to make problem~\eqref{eq:dual_restricted} amenable for a distributed
solution, we enforce a sparsity structure that matches the network.  To this
end, we introduce copies of the common optimization variable $\smu{}$ and we
copy also its domain. Moreover, we enforce coherence constraints among the
copies $\smu{i}$ having the sparsity of the connected graph $\GG$, thus obtaining
\begin{align}
\begin{split}
  \max_{\smu{1}, \ldots, \smu{N} } \: & \: \sum_{i=1}^N q_i (\smu{i} )
  \\
  \subj \: & \: \smu{i} \succeq \0, \smu{i}^\top \1 \le M, \hspace{0.30cm} i \in\until{N}
  \\[1.2ex]
  & \: \smu{i} = \smu{j}, \hspace{1.75cm} (i,j) \in \EE.
\end{split}
\label{eq:restricted_dual_problem_copies}
\end{align}
Being problem~\eqref{eq:restricted_dual_problem_copies} an equivalent version of
problem~\eqref{eq:dual_restricted}, it has the same optimal cost $q^\star = f^\star$. 

On problem~\eqref{eq:restricted_dual_problem_copies} we would like to use a dual
decomposition approach with the aim of obtaining a distributed algorithm.  That
is, the leading idea is to derive the dual of
problem~\eqref{eq:restricted_dual_problem_copies}\! and apply a subgradient method
to solve it.

We start deriving the dual problem of~\eqref{eq:restricted_dual_problem_copies}
by dualizing only the coherence constraints. Consider the partial
Lagrangian
\begin{align}
  \LL (\smu{1},\ldots,  \smu{N}, \sLambda ) 
 = \!
 \sum_{i=1}^N \! \Big( q_i( \smu{i} ) 
 + \!\!
 \sum_{j\in\nbrs_i} \! \slambda{ij}^{\!\top} (\smu{i} - \smu{j} ) \! \Big),
\label{eq:lagrangian_dual_copies}
\end{align}
where $\sLambda\in \real^{S\cdot |\EE|}$ is the vector stacking each Lagrange
multiplier $\slambda{ij} \in \real^S$, with $(i,j)\in \EE$, associated to the
constraint $\smu{i} - \smu{j} = \0$. 
Notice that we have not dualized the local constraints
$\{ \smu{i} \succeq \0 \mid \smu{i}^\top \1 \le M\}$.

Since the communication graph $\GG$ is undirected, we can exploit
the symmetry of the constraints. Indeed, for each $(i,j)\in\EE$ we also have
$(j,i) \in\EE$, and, expanding all the terms in~\eqref{eq:lagrangian_dual_copies}, 
for given $i$ and $j$, we always have both
the terms $\slambda{ij}^\top (\smu{i} - \smu{j} )$ and
$\slambda{ji}^\top (\smu{j} - \smu{i} )$.
Thus, after some simple algebraic manipulations, we can 
rephrase~\eqref{eq:lagrangian_dual_copies} as
$\LL(\smu{1}, \ldots, \smu{N}, \sLambda ) = 
  \sum_{i=1}^N \Big( q_i( \smu{i} ) + \smu{i}^\top \sum_{j\in\nbrs_i} (\slambda{ij} - \slambda{ji})  \Big)$,
which is separable with respect to $\smu{i}$, $i\in\until{N}$.
Thus, the dual function of~\eqref{eq:restricted_dual_problem_copies} is
\begin{align}
\begin{split}
  \eta( \sLambda ) & \!=
  \sum_{i=1}^N \eta_i \big( \{ \slambda{ij},\slambda{ji} \}_{j\in\nbrs_i} \big),
\end{split}
\label{eq:eta_definition}
\end{align}
where, for all $i\in\until{N}$,
\begin{align}
  \! \! \! \eta_i \big( \! \{ \slambda{ij},\slambda{ji} \}_{j\in\nbrs_i} \big) \! = \!\!\!
    \sup_{  \substack{ \smu{i} \succeq \0, \\ \smu{i}^\top \1 \le M } } \!\!\!
      \Big( \! q_i (\smu{i})   \! + \! 
        \smu{i}^\top \! \!  \sum_{j\in\nbrs_i} \!\! (\slambda{ij}  \! -  \!  \slambda{ji} )
    \! \Big).
\label{eq:eta_i_definition}
\end{align}
Finally, by denoting the domain of $\eta$ as
  $\domainEta = \{ \sLambda \in \real^{S \cdot |\EE|} \mid \eta( \sLambda ) < +\infty \}$,
the dual of problem~\eqref{eq:restricted_dual_problem_copies} reads
\begin{align}
\label{eq:dual_dual}
  \min_{ \sLambda \in \domainEta } \: \eta ( \sLambda  ) = 
  \min_{ \sLambda \in \domainEta } \:
  \sum_{i=1}^N
  \eta_i \big( \{\slambda{ij},\slambda{ji}\}_{j\in\nbrs_i} \big).
\end{align}

Since problem~\eqref{eq:dual_dual} is a dual program, then it is a 
convex (constrained) problem. 
Moreover, its cost 
function $\eta ( \sLambda  )$ is very structured since it is a sum of contributions 
$\eta_i$ and each of them depends only on neighboring variables.
In the next lemma we characterize the domain of
problem~\eqref{eq:dual_dual}.
\begin{lemma}
\label{lem:dual_dual_domain}
  The domain $\domainEta$ of $\eta$ in~\eqref{eq:eta_definition} is
  $\real^{S \cdot |\EE|}$, thus optimization
  problem~\eqref{eq:dual_dual} is unconstrained.
\end{lemma}
\begin{proof}
  We show that each $\eta_i \big( \{\slambda{ij},\slambda{ji}\}_{j\in\nbrs_i} \big)$ is finite for
  all $\{\slambda{ij},\slambda{ji}\}_{j\in\nbrs_i}$.
  Each function $q_i(\smu{i})$ is concave on its domain $\smu{i}\succeq \0$ 
  for all $i\in\until{N}$. In fact,
  from the definition of $q_i$ in~\eqref{eq:qi_definition}, we notice
  that it is obtained as minimization over a nonempty compact set $X_i$ of the function 
  $f_i ( \sx{i} ) + \smu{i}^\top \bg_i (\sx{i})$. Such a function is concave (in fact linear) in $\smu{i}$, thus, 
  following the proof of \cite[Proposition~5.1.2]{bertsekas1999nonlinear}, we can conclude that every
  $q_i$ is concave over its convex domain $\smu{i}\succeq \0$.
  For each $i\in\until{N}$, the function $\eta_i$ as
  defined in~\eqref{eq:eta_i_definition} is obtained by maximizing a (concave)
  continuous function ($q_i$ plus a linear term) over a compact set 
  and, thus,
  has always a finite value, so that the proof follows. %
\end{proof}

It is worth noting that Lemma~\ref{lem:dual_dual_domain} strongly
relies on the compactness of
$\{ \smu{i} \succeq \0 \mid \smu{i}^\top \1 \le M\}$. This means that without
the primal relaxation, $\domainEta$ is not guaranteed to be
$\real^{S \cdot |\EE|}$. In Section~\ref{subsec:on_the_relaxation}, we better
clarify this aspect.

Next we characterize the optimization problem~\eqref{eq:dual_dual}.
\begin{lemma} \label{lem:dual_dual}
  Let $M$ be such that $M > \| \smu{}^\star\|_1$ with
  $\smu{}^\star$ an optimal solution of the dual of problem~\eqref{eq:primal_original}.
  Problem~\eqref{eq:dual_dual} has a
  bounded optimal cost, call it $\eta^\star$, and
  a nonempty optimal solution set.
  Moreover, it enjoys strong duality with~\eqref{eq:restricted_dual_problem_copies}.
  Also, it holds $\eta^\star = f^\star$,
  where $f^\star$ is the optimal solution of the original primal problem~\eqref{eq:primal_original}.
\end{lemma}
\begin{proof}
  Since~\eqref{eq:restricted_dual_problem_copies} is equivalent
  to~\eqref{eq:dual_original}, then by
  Lemma~\ref{lem:relaxation_equivalence} its optimal cost is finite and equal to
  $q^\star$. 
  Since each $q_i$ is concave
  as shown in the proof of Lemma~\ref{lem:dual_dual_domain}, then also
  $\sum_{i=1}^N q_i(\smu{i})$ is a concave function of
  $(\smu{1}, \ldots,\smu{N})$. Thus, since the domain
  $\{ \smu{1}, \ldots,\smu{N} \mid \smu{i} \succeq \0 \mid \smu{i}^\top \1 \le M, i\in\until{N}\}$,
  is polyhedral, by
  \cite[Proposition~5.2.1]{bertsekas1999nonlinear} 
  strong duality between problem~\eqref{eq:restricted_dual_problem_copies} and 
  its dual~\eqref{eq:dual_dual} holds, i.e.,
  $\eta^\star$ is finite since it holds $\eta^\star = q^\star$.
  From the same proposition, we have that
  the optimal solution set of~\eqref{eq:dual_dual} is nonempty.
  The equality $\eta^\star = f^\star$ follows
  readily by strong duality between problems~\eqref{eq:primal_original} 
  and~\eqref{eq:dual_original}, which concludes the proof.
\end{proof}

\subsection{Distributed Subgradient Method}
\label{subsec:distributed_subgradient_method}
We detail in this subsection how to explicitly design a distributed dual decomposition
algorithm to solve problem~\eqref{eq:dual_restricted} based on a 
subgradient iteration applied to problem~\eqref{eq:dual_dual}.

Exploiting the separability of $\eta$ in~\eqref{eq:eta_definition}, we recall
how to compute each component of a subgradient
of $\eta$ at a given $\sLambda  \in \real^{S \cdot |\EE|}$,
see e.g., \cite[Section~6.1]{bertsekas1999nonlinear} 
That is, it holds
\begin{align}
\frac{\tilde \partial
   \eta ( \sLambda ) }{ \partial \slambda{ij} }
  = \smu{i}^\star - \smu{j}^\star,
\label{eq:eta_subgradient}
\end{align}
where $\frac{\tilde \partial \eta ( \cdot )}{\partial \slambda{ij}}$ denotes the component
associated to the variable $\slambda{ij}$ of a subgradient of $\eta$, and
\begin{align}
  \smu{k}^\star \in \argmax_{ \smu{k}\succeq \0, \smu{k}^\top \1 \le M} 
  \Big( 
    q_k( \smu{k} ) +
    \smu{k}^\top \sum_{h\in\nbrs_k} (\slambda{kh} - \slambda{hk}) 
  \Big),
\label{eq:mustar_k}
\end{align}
for $k=i,j$.

Having recalled how to compute subgradients of $\eta$, we are ready to summarize
how the subgradient method reads when applied to problem~\eqref{eq:dual_dual}.
At each iteration $t$, each node $i\in\until{N}$:
\begin{enumerate}[label=(S\arabic*),ref=(S\arabic*)]
\item\label{subgrad_alg_s1} receives $\slambda{ji}^t$, $j \in \nbrs_i$, and
  computes $\smu{i}^{t+1}$ as an optimal solution of
  \begin{align}
    \max_{ \smu{i} \succeq \0, \smu{i}^\top \1 \le M} \Big ( q_i( \smu{i} ) +
    \smu{i}^\top \!\! \sum_{j\in\nbrs_i} (\slambda{ij}^t - \slambda{ji}^t ) \Big );
    \label{eq:dual_dual_subgradient}
  \end{align}
  
\item\label{subgrad_alg_s2} receives the updated $\smu{j}^{t+1}$, 
  $j \in \nbrs_i$ and updates $\slambda{ij}$, $j\in\nbrs_i$, via
  \begin{align*}
    \slambda{ij}^{t + 1} = \slambda{ij}^t - \gamma^t (\smu{i}^{t +1} -\smu{j}^{t + 1}),
  \end{align*}
  where $\gamma^t$ is the step-size.
\end{enumerate}

Notice that~\ref{subgrad_alg_s1}--\ref{subgrad_alg_s2} is a distributed algorithm, i.e., 
it can be implemented by means of local computations and communications without 
any centralized step.
However, we want to stress that the algorithm is \emph{not} implementable as it is written, 
since the functions $q_i$ are still not explicitly available. %

The next lemma states a property on the subdifferential of the convex function $\eta$.

\begin{lemma}
\label{lem:bounded_subgradient}
  The subgradients of $\eta$ are uniformly bounded, i.e.,
  there exists a positive constant $C$ such that
  for every $\sLambda \in \real^{S\cdot|\EE|}$ with components
  $\slambda{ij} \in \real^S$, $\forall \, (i,j)\in \EE$, any subgradient $\tilde \partial
  \eta ( \sLambda ) / \partial \sLambda $ satisfies
    $\| \frac{\tilde \partial \eta ( \sLambda ) }{ \partial \sLambda } \| \le C$.
\end{lemma}
\begin{proof}
  To prove the lemma, we show that each component 
  $\frac{\tilde \partial \eta ( \sLambda ) }{ \partial \slambda{ij} }$
  of $\frac{\tilde \partial \eta ( \sLambda ) }{ \partial \sLambda }$
  is bounded.
  Using~\eqref{eq:eta_subgradient}, it is sufficient to show that
  ${\smu{i}^\star}$ and ${\smu{j}^\star}$, associated to the given $\sLambda$, 
  are uniformly bounded and, hence, their difference.
  Since, from equation~\eqref{eq:mustar_k} the two are obtained as maxima of a
  concave function over a compact domain, the proof follows.%
\end{proof}

\section{Algorithm Analysis: Convergence Proof}
\label{sec:analysis}
This section is devoted to prove the convergence of the \algacronym/ distributed algorithm
formally stated in Theorem~\ref{thm:convergence}. 

\subsection{Preparatory Results} %
\label{subsec:intermediate_results}
We give two intermediate results that represent building
blocks for the convergence proof given in Section~\ref{subsec:convergence_proof}.
The next lemma is instrumental to the second one.
\begin{lemma} \label{lem:inner_maximization_dual}
Consider the optimization problem
\begin{align}
\begin{split}
  \max_{\smu{i} } \: & \:
     f_i (\sx{i}) + \smu{i}^{\top}\Big(
        \bg_i (\sx{i}) + \sum_{j\in\nbrs_i} ( \slambda{ij}^t - \slambda{ji}^t ) \Big)
  \\
  \subj \: & \: \smu{i}\succeq \0,  \smu{i}^\top \1 \le M,
\end{split}
\label{eq:inner_maximization_problem}
\end{align}
with given $\sx{i}$, $\slambda{ij}^t$ and $\slambda{ji}^t$,
$j\in\nbrs_i$, and $M>0$.
Then, its dual problem is
\begin{align}
\begin{split}
  \min_{\rho_{i}} \: & \: f_i (\sx{i}) + M \rho_{i}
  \\
  \subj \:
   & \: \rho_{i} \ge 0
  \\
   & \: \bg_i( \sx{i} ) +\sum_{j\in\nbrs_i} (\slambda{ij}^t - \slambda{ji}^t ) \preceq \rho_{i} \1,
\end{split}
\label{eq:rho_formulation}
\end{align}%
and strong duality holds.
\end{lemma}
\begin{proof}
  First, since $\sx{i}$, $\slambda{ij}^t$ and $\slambda{ji}^t$ are given,
  problem~\eqref{eq:inner_maximization_problem} is a feasible \emph{linear}
  program (the box constraint is nonempty) with compact domain.  Thus, both
  problem~\eqref{eq:inner_maximization_problem} and its dual have finite optimal
  cost and strong duality holds.

  In order to show that~\eqref{eq:rho_formulation} is the dual of~\eqref{eq:inner_maximization_problem}, 
  we introduce a multiplier $\rho_{i} \ge 0$ associated to the constraint $ M - \smu{i}^\top \1 \ge 0$.
  Then the dual function of~\eqref{eq:inner_maximization_problem}
  is defined as $\max_{ \smu{i} \succeq 0 } f_i (\sx{i}) + M \rho_{i}
    + \smu{i}^\top \Big( \bg_i( \sx{i} ) 
      +\!\!\sum_{j\in\nbrs_i} (\slambda{ij}^t - \slambda{ji}^t ) -\rho_{i}\1 \Big)$.
  It is equal to $f_i (\sx{i}) + M \rho_{i}$ if 
  $\bg_i( \sx{i} ) +\sum_{j\in\nbrs_i} (\slambda{ij}^t - \slambda{ji}^t ) - \rho_{i} \1 \preceq \0$ 
  and $+\infty$ otherwise. 
  Finally, the minimization of the dual function
  on its domain with respect to $\rho_{i} \ge 0$ gives problem~\eqref{eq:rho_formulation} and 
  concludes the proof. %
\end{proof}

In the following, we propose a technique to make step~\eqref{eq:dual_dual_subgradient}
explicit.
By plugging in~\eqref{eq:dual_dual_subgradient} the definition of $q_i$, given
in~\eqref{eq:qi_definition}, the following max-min optimization problem is
obtained:
\begin{align}
  \max_{ \substack{ \smu{i} \succeq \0, \\ \smu{i}^\top \1 \le M }}
  \min_{\sx{i} \in X_i} \!\!
  \bigg( f_i ( \sx{i} ) \!+\! \smu{i}^\top \! \Big( 
    \bg_i( \sx{i} ) \!+ \! \!\! \sum_{j\in\nbrs_i} \! ( \slambda{ij}^t \!-\! \slambda{ji}^t ) 
  \! \Big) \!\! \bigg).
\label{eq:maxmin}
\end{align}

The next lemma allows us to recast problem~\eqref{eq:maxmin} in a more convenient
formulation from a computational point of view.

\begin{lemma}
  Consider the optimization problem
  \begin{align}
    \begin{split}
    \min_{ \sx{i}, \rho_{i} } \: & \: f_i (\sx{i}) + M \rho_{i}
    \\
    \subj \: & \: \rho_{i} \ge 0, \:\: \sx{i} \in X_i
    \\
    & \: 
    \bg_i (\sx{i} ) + \sum_{j\in\nbrs_i}  \big( \slambda{ij}^t - \slambda{ji}^t \big) \preceq \rho_{i}\1.
    \end{split}
  \label{eq:alg_minimization_lemma}
  \end{align}
  A finite primal-dual optimal solution pair of~\eqref{eq:alg_minimization_lemma}, call it
  $\big( ( \sx{i}^{t+1}, \rho_{i}^{t+1} ), \smu{i}^{t+1} \big)$, does exist and
  $( \sx{i}^{t+1}, \smu{i}^{t+1} )$ is a solution of~\eqref{eq:maxmin}.
\label{lem:dual_minmax_equivalence}
\end{lemma}
\begin{proof}
  Problem~\eqref{eq:alg_minimization_lemma} is a feasible convex program, in fact 
  $f_i (\sx{i}) + M \rho_{i}$ is convex, 
  the set $X_i$ is nonempty, convex and compact,
  the constraint $\rho_{i} \ge 0$ is convex as well as the inequality constraint 
  $\bg_i (\sx{i} ) + \sum_{j\in\nbrs_i} \big( \slambda{ij}^t - \slambda{ji}^t \big) - \rho_{i}\1 \preceq \0$.
  Then, by choosing a sufficiently large $\rho_{i}$, we can show that the Slater's 
  constraint qualification is satisfied and, thus, strong duality holds. 
  Therefore, a primal-dual optimal solution pair
  $(\sx{i}^{t+1}, \rho_{i}^{t+1},\smu{i}^{t+1} )$ of~\eqref{eq:alg_minimization_lemma} exists. 
  Moreover, problem~\eqref{eq:alg_minimization_lemma} can be recast as
  \begin{align*}
    \min_{ \sx{i} \in X_i } \!\! \Bigg(\!
    \min_{ \rho_{i} \ge 0,  \: \bg_i (\sx{i} ) +\sum\limits_{j\in\nbrs_i}
      (\slambda{ij}^t - \slambda{ji}^t ) \preceq \rho_{i}\1  }
      f_i ( \sx{i} ) + M \rho_{i} \!
    \Bigg) \! .
  \end{align*}
  
  By Lemma~\ref{lem:inner_maximization_dual}, we can substitute 
  the inner minimization with its equivalent dual maximization obtaining
  \begin{align}
    \min_{\sx{i} \in X_i}
    \max_{\substack{ \smu{i} \succeq \0, \\ \smu{i}^\top \1 \le M }} \!\!
    \bigg( \! f_i ( \sx{i} ) \!+\! \smu{i}^\top \! \Big(
      \bg_i( \sx{i} ) \!+ \! \! \sum_{j\in\nbrs_i} \! ( \slambda{ij}^t \!-\! \slambda{ji}^t ) 
    \! \Big) \!\! \bigg).
  \label{eq:saddle_point_problem}
  \end{align}  

  Let
  $\phi(\sx{i},\smu{i}) \! = \! 
    f_i ( \sx{i} ) \!+\! \smu{i}^\top \! \Big( 
    \bg_i( \sx{i} ) \!+ \! \sum_{j\in\nbrs_i} \! ( \slambda{ij}^t \!-\! \slambda{ji}^t ) 
  \! \Big)$
  and observe that (i) $\phi(\cdot,\smu{i})$ is closed and convex for all
  $\smu{i} \succeq \0$ (affine transformation of a convex function with compact
  domain $X_i$) and (ii) $\phi(\sx{i}, \cdot )$ is closed and concave since it is
  a linear function with compact domain 
  ($\{ \smu{i} \succeq \0 \mid \smu{i}^\top \1 \le M\}$), for
  all $\sx{i}\in\real^S$.
  Thus, we can invoke \cite[Propositions~4.3]{bertsekas2009min}
  to switch $\min$ and $\max$ operators in~\eqref{eq:saddle_point_problem}, and write
  \begin{align}
  \begin{split}
    &
    \min_{\sx{i} \in X_i}
    \max_{ \substack{ \smu{i} \succeq \0, \\ \smu{i}^\top \1 \le M }}
    \!\!
    \bigg( \! f_i ( \sx{i} ) \!+\! \smu{i}^\top \! \Big(
      \bg_i( \sx{i} ) \!+ \! \! \sum_{j\in\nbrs_i} \!\! ( \slambda{ij}^t \!-\! \slambda{ji}^t ) 
    \! \Big) \!\! \bigg)
    \\
    & = \! \! \!
    \max_{ \substack{ \smu{i} \succeq \0, \\ \smu{i}^\top \1 \le M }}
    \min_{\sx{i} \in X_i} \!\!
      \bigg( \! f_i ( \sx{i} ) \!+\! \smu{i} ^\top \! \Big( 
      \bg_i( \sx{i} ) \!+ \! \! \sum_{j\in\nbrs_i} \! \! ( \slambda{ij}^t \!-\! \slambda{ji}^t ) 
    \! \Big) \!\! \bigg).
    \end{split}
  \label{eq:minmax}
  \end{align}
  which is~\eqref{eq:maxmin}, thus concluding the proof. %
\end{proof}

We highlight that problem~\eqref{eq:alg_minimization_lemma} is the local optimization 
step~\eqref{eq:alg_minimization} in the \algacronym/ distributed algorithm.

Finally, the next corollary makes a connection between the optimal cost of $i$-th
problem~\eqref{eq:alg_minimization_lemma} and the value of the $i$-th local term
$\eta_i$ (defined in \eqref{eq:eta_i_definition}) of the second dual function
$\eta$ (defined in \eqref{eq:eta_definition}).
\begin{corollary}
  Let  $( \sx{i}^{t+1}, \rho_{i}^{t+1} )$ be a solution of~\eqref{eq:alg_minimization_lemma}
  with given $\slambda{ij}^t$ and $\slambda{ji}^t$ for $j\in\nbrs_i$. Then
  \begin{align}
    \eta_i \big( \{ \slambda{ij}^t ,\slambda{ji}^t \}_{j\in\nbrs_i} \big) 
    = f_i ( \sx{i}^{t+1} ) + M \rho_{i}^{t+1},
    \label{eq:optimal_eta_frho}
  \end{align}
  with $\eta_i$ defined in~\eqref{eq:eta_i_definition}.
  \label{cor:optimal_eta_frho}
\end{corollary}
\begin{proof}
  In the proof of Lemma~\ref{lem:dual_minmax_equivalence}, we have shown that
  condition~\eqref{eq:minmax} holds for all $t \ge 0$. Its left hand side has
  optimal cost $f_i ( {\sx{i}}^{t+1} ) + M\rho_{i}^{t+1} $, while the one of the
  right hand side is exactly the definition of
  $\eta_i \big( \{ \slambda{ij}^t ,\slambda{ji}^t \}_{j\in\nbrs_i} \big) $
  in~\eqref{eq:eta_i_definition}. Thus, equation~\eqref{eq:minmax} can be
  rewritten as
  \begin{align*}
    f_i ( \sx{i}^{t+1} ) + M\rho_{i}^{t+1} =
    \eta_i \big( \{ \slambda{ij}^t ,\slambda{ji}^t \}_{j\in\nbrs_i} \big),
  \end{align*}
  for all $i\in\until{N}$, concluding the proof.%
\end{proof}

\subsection{Proof of Theorem~\ref{thm:convergence}}
\label{subsec:convergence_proof}

  To prove statement \emph{(i)}, we show that the
  \algacronym/ distributed algorithm is an operative way to implement the
  subgradient method \ref{subgrad_alg_s1}--\ref{subgrad_alg_s2}
  and, that \ref{subgrad_alg_s1}--\ref{subgrad_alg_s2} solves
  problem~\eqref{eq:dual_restricted}.
  
  First, let $\{ \smu{i}^t \}_{t\ge 0}$ and 
  $\{\slambda{ij}^t \}_{t\ge 0}$, $j \in \nbrs_i$,
  be the auxiliary sequences generated by the \algacronym/ distributed 
  algorithm associated to 
  $\{ ( \sx{i}^t , \rho_{i}^t ) \}_{t\ge 0}$,
  for each $i\in \until{N}$.
  From Lemma~\ref{lem:dual_minmax_equivalence}, 
  a primal-dual optimal solution pair 
  $\big( ( \sx{i}^{t}, \rho_{i}^{t} ), \smu{i}^{t} \big)$ 
  of~\eqref{eq:alg_minimization} in fact exists at each iteration $t$, 
  so that the algorithm is well-posed.
  Second, to show that \algacronym/ implements
  \ref{subgrad_alg_s1}--\ref{subgrad_alg_s2} we notice that
  update~\eqref{eq:alg_update} and~\ref{subgrad_alg_s2} are trivially identical.
  As for \ref{subgrad_alg_s1}, we have shown in the discussion after
  Lemma~\ref{lem:inner_maximization_dual}, that equation~\eqref{eq:maxmin} is
  an explicit expression for~\eqref{eq:dual_dual_subgradient} in
  \ref{subgrad_alg_s1}. Thus, by invoking
  Lemma~\ref{lem:dual_minmax_equivalence}, we can conclude that finding the dual
  part of a primal-dual optimal solution pair of~\eqref{eq:alg_minimization}
  corresponds to performing~\ref{subgrad_alg_s1}. Therefore, the sequences
  $\{\slambda{ij}^t\}_{t\ge 0}$, $(i,j) \in \EE$ generated by \algacronym/ and
  by \ref{subgrad_alg_s1}--\ref{subgrad_alg_s2} coincide.
  Third, we show that \algacronym/ solves problem~\eqref{eq:dual_dual}. 
  By Lemma~\ref{lem:dual_dual} the optimal solution set of~\eqref{eq:dual_dual} is nonempty 
  and by Lemma~\ref{lem:bounded_subgradient} the subgradients of $\eta$ are uniformly bounded.
  Since the step-size $\gamma^t$ satisfies Assumption~\ref{ass:step-size}, we can
  invoke~\cite[Proposition~3.2.6]{bertsekas2015convex}
  to conclude that the 
  sequence $\{\slambda{ij}^t \}_{t\ge 0}$, $(i,j) \in \EE$ generated by
  \algacronym/ (or equivalently by \ref{subgrad_alg_s1}--\ref{subgrad_alg_s2})
  converges to an optimal solution of~\eqref{eq:dual_dual}.
  Then,
  we use~\eqref{eq:optimal_eta_frho} in Corollary~\ref{cor:optimal_eta_frho} and
  take the limit as $t\to\infty$, thus obtaining
  \begin{align*}
    \lim_{t\to\infty} \sum_{i=1}^N \! \Big(f_i ( \sx{i}^{t+1} ) + M \rho_{i}^{t+1} \Big)
    & = \lim_{t\to\infty} \sum_{i=1}^N \! \eta_i \big( \{
      \slambda{ij}^t ,\slambda{ji}^t \}_{j\in\nbrs_i} \big) 
\\
&    = \eta^\star = f^\star,
  \end{align*}
  where the last equality follows by Lemma~\ref{lem:dual_dual}, so that the proof 
  of the first statement is complete. 
  To prove statement~\emph{(ii)}, i.e., the primal recovery property, we
  start by studying the properties of the aggregated vector
  $(\sx{1}^t,\ldots,\sx{N}^t,\rho_{1}^t,\ldots,\rho_{N}^t)$.
  By construction, for all $t\ge0$ each pair $(\sx{i}^t,\rho_{i}^t)$ satisfies $\sx{i}^t \in X_i$,
  $\rho_i^t \ge 0$ and
  $    \bg_i ( \sx{i}^t ) + \sum_{j\in\nbrs_i} \big( \slambda{ij}^{t-1} - \slambda{ji}^{t-1} \big)
    \preceq \rho_{i}^t\1$.
  Summing over $i\in\until{N}$ the previous condition, it follows that
  \begin{align}
    \sum_{i=1}^N \bg_i(\sx{i}^t ) + \sum_{i=1}^N 
    \sum_{j\in\nbrs_i} \big( \slambda{ij}^{t-1} - \slambda{ji}^{t-1} \big)
    \preceq \sum_{i=1}^N \rho_{i}^t \1,
  \label{eq:proof_coupling_with_lambda}
  \end{align}
  for all $t\ge 0$.
  Let us denote by $a_{ij}$ the $(i,j)$-th entry of the adjacency matrix associated
  to the undirected graph $\GG$. Then, we can write
  \begin{align*}
    \begin{split}
    \sum_{i=1}^N
    \sum_{j\in\nbrs_i} (\slambda{ij}^t - \slambda{ji}^t )
    & \stackrel{(a)}{=}
    \sum_{i=1}^N \sum_{j=1}^N a_{ij}(\slambda{ij}^t - \slambda{ji}^t )
    \\
    & %
    = \sum_{i=1}^N \sum_{j=1}^N a_{ij} \slambda{ij}^t
    - 
    \sum_{i=1}^N \sum_{j=1}^N a_{ij} \slambda{ji}^t
    \stackrel{(b)}{=} 0,    
    \end{split}
  \end{align*}
  where (a) follows by writing the sum over neighboring agents in terms of the
  adjacency matrix, while (b) holds since $\GG$ is undirected (so that
  $a_{ij} = a_{ji}$ for all $(i,j)\in\EE$), which implies that the two
  summations in the second line are identical for all $t\ge 0$.
  Hence, equation~\eqref{eq:proof_coupling_with_lambda} reduces to
  \begin{align}
    \sum_{i=1}^N \bg_i ( \sx{i}^t )
    \preceq \sum_{i=1}^N \rho_{i}^t \1,
  \label{eq:proof_aggregate_feasibility_condition}
  \end{align}
  for all $t\ge 0$.
  Equation~\eqref{eq:proof_aggregate_feasibility_condition} shows that for all
  $t\ge 0$ the aggregate vector 
  $(\sx{1}^t,\ldots,\sx{N}^t,\rho_{1}^t,\ldots,\rho_{N}^t)$
  is feasible for the following optimization problem
  \begin{align}
  \begin{split}
    \min_{\substack{ \sx{1},\ldots,\sx{N} \\ \rho_{1},\ldots,\rho_{N} }}
    \: & \: \sum_{i=1}^N f_i ( \sx{i} ) + M \sum_{i=1}^N \rho_{i}
    \\[0.2cm]
    \subj \: & \: \rho_{i} \ge 0, \:\: \sx{i} \in X_i, \:\: i \in\until{N}
    \\
    & \: \sum_{i=1}^N \bg_i (\sx{i}) \preceq \sum_{i=1}^N \rho_{i} \1.
  \end{split}
  \label{eq:proof_problem}
  \end{align}
  
  Notice that, by defining $\rho = \sum_{i=1}^N \rho_{i}$, 
  problem~\eqref{eq:proof_problem} is equivalent to problem~\eqref{eq:primal_relaxed}. 
  Thus, at each iteration $t$ the point
  $(\sx{1}^t,\ldots,\sx{N}^t,\sum_{i=1}^N \rho_{i}^t)$ is feasible for
  problem~\eqref{eq:primal_relaxed}. The equivalence also shows
  that $\rho_{i}$ is not a copy of $\rho$, but it is the $i$-th contribution 
  to $\rho$.
  
  We now show that every limit point of the sequence 
  $\{\sx{1}^t,\ldots,\sx{N}^t, \sum_{i=1}^N \rho_{i}^t\}_{t\ge0}$ 
  is feasible for problem~\eqref{eq:primal_relaxed}. %
  By construction, each $\sx{i}^t \in X_i$ for all $i\in\until{N}$, so that $\{\sx{i}^t\}_{t\ge0}$ is bounded. 
  Moreover, from the statement \emph{(i)} of the theorem, also the sequence $\{ \sum_{i=1}^N \rho_{i}^t\}_{t\ge0}$ is 
  bounded since $\{ \sum_{i=1}^N f_i ( \sx{i}^t ) + M \sum_{i=1}^N \rho_{i}^t\}_{t\ge 0}$ converges
  to a finite value $f^\star$.
  Since the sequence of vectors
  $\{ (\sx{1}^t,\ldots,\sx{N}^t,\sum_{i=1}^N \rho_{i}^t) \}_{t\ge 0}$ is
  bounded, then there exists a sub-sequence of indices
  $\{ t_n \}_{n\ge 0} \subseteq \{ t\}_{t\ge 0}$ such that the sub-sequence
  $\{ (\sx{1}^{t_n},\ldots,\sx{N}^{t_n},\sum_{i=1}^N \rho_{i}^{t_n}) \}_{n \ge
    0}$
  converges to a limit point $( \barsx{1},\ldots, \barsx{N}, \bar{\rho} )$.
  From the first statement of the theorem we have that
  $( \barsx{1},\ldots, \barsx{N}, \bar{\rho} )$ satisfies 
  \begin{align*}
    \sum_{i=1}^N f_i ( \barsx{i}) + M \bar{\rho} = f^\star.
  \end{align*}
  Moreover, since each component of $\bg_i$ is a (finite) convex function over $\real^{n_i}$, it is also 
  continuous over any compact subset of $\real^{n_i}$. Thus, by taking the limit
  as $n \to \infty$ in~\eqref{eq:proof_aggregate_feasibility_condition} with $t= t_n$,
  it also holds
  \begin{align}
    \sum_{i=1}^N \bg_i(\barsx{i}) \preceq \bar{\rho} \, \1.
    \label{eq:asymptotic_coupling_constraint}
  \end{align}
  By Proposition~\ref{prop:relaxed_solution_structure} it must hold that
  $(\barsx{1}, \ldots, \barsx{N}, \bar{\rho}) = (\barsx{1}, \ldots, \barsx{N},
  0)$,
  i.e., $\bar{\rho} = 0$.  Thus, \eqref{eq:asymptotic_coupling_constraint} holds
  with $\bar{\rho} = 0$ and, thus, guarantees that every limit point of
  $(\sx{1}^t,\ldots,\sx{N}^t)$ is feasible for the (not relaxed) coupling
  constraint in the original problem~\eqref{eq:primal_original} and thus optimal
  for that problem.  So that the proof follows.

\subsection{Discussion on the Necessity of the Relaxation}
\label{subsec:on_the_relaxation}

In this subsection we show how our approach reads when no dual restriction 
is applied. This will further highlight the strength of the proposed strategy.

Suppose we do not restrict the original dual problem~\eqref{eq:dual_original}, but we still apply the same formal 
derivation given in the previous sections.
Then, the counterpart of~\eqref{eq:eta_definition} is
$ \eta^{ \texttt{NR} }( \sLambda ) 
  = 
  \sum_{i=1}^N \eta_i^{ \texttt{NR} } \big( \{ \slambda{ij},\slambda{ji} \}_{j\in\nbrs_i} \big)
$
with 
$$
  \eta_i^{ \texttt{NR} } \big( \{ \slambda{ij},\slambda{ji} \}_{j\in\nbrs_i} \big) 
  =
  \sup_{ \smu{i} \succeq \0 } 
    \Big( q_i (\smu{i}) 
    +
    \smu{i}^{\!\top}  \sum_{j\in\nbrs_i} \! (\slambda{ij}  \! -  \!  \slambda{ji} )
  \Big),
$$
for all $i\in\until{N}$.
Finally, by denoting the domain of $\eta^{ \texttt{NR} } $ as $\domainEtaRelaxed$,
we have that the counterpart of problem~\eqref{eq:dual_dual} is
$\min_{ \sLambda \in D_{ \sLambda}^{ \texttt{NR} }  } \eta^{ \texttt{NR} } ( \sLambda )$.
We notice that such a problem is a constrained minimization since,
differently from the relaxed case, the domain $\domainEtaRelaxed$
does not always coincide with the entire space $\real^{S \cdot |\EE|}$
(Cf. Lemma~\ref{lem:dual_dual_domain}).  Thus, to apply the subgradient method
we need to adapt~\ref{subgrad_alg_s1}--\ref{subgrad_alg_s2} by appending an
additional projection step, i.e.,
$\sLambda^{t+1} = \big[ \tildesLambda{}^{t+1} \big]_{ \domainEtaRelaxed  }$,
where each component $\tildeslambda{ij}^{t+1}$ of $\tildesLambda{}^{t+1}$ is the 
result of \ref{subgrad_alg_s2} and
$[\,\cdot\,]_{\domainEtaRelaxed }$ denotes the Euclidean projection
onto $\domainEtaRelaxed$. Notice that the projection onto
$\domainEtaRelaxed$ of the entire $\tildesLambda{}^{t+1}$ prevents
the distributed implementation of the algorithm.

It is worth noting that, being the set $\{ \smu{i}\in\real^S \mid \smu{i}\succeq \0 \}$ not
compact, then Lemma~\ref{lem:dual_minmax_equivalence} would not hold.
The theoretical issue is that switching $\min$ and $\max$ operators
in~\eqref{eq:minmax} may not be possible due to the non-compact domains
(Cf. \cite[Propositions~4.3]{bertsekas2009min}.
Moreover, differently from the relaxed case, 
Lemma~\ref{lem:bounded_subgradient} does not hold anymore so that
no guarantees about the boundedness of subgradients of $\eta^{\texttt{NR}}$ can 
be established. Thus, the convergence result about the centralized subgradient 
method cannot be invoked.

\section{Application to Distributed Microgrid Control}
\label{sec:simulations}

In this section we present a computational study of our \algacronym/ distributed
algorithm tailored to an instance of a distributed MPC scheme for microgrid
control.

\subsection{Microgrid Model}

A microgrid consists of generators, controllable loads, storage devices and a connection to the main grid \cite{zamora2010controls}.
Generators are collected in the set $\GEN$. At each time instant $\tau$ in a
given horizon $[0,T]$, they generate power $p_{\gen,i}^{\tau}$
that must satisfy magnitude and rate bounds, i.e., 
for all $i \in \GEN$, $\ubar{p} \le p_{\gen,i}^{\tau} \le \bar{p}$, with
$\tau \in [0,T]$, and
$\ubar{r} \le p_{\gen,i}^{\tau+1} - p_{\gen,i}^{\tau} \le \bar{r}$, with
$\tau \in [0,T-1]$, 
for given positive scalars
$\ubar{p}$, $\bar{p}$, $\ubar{r}$ and $\bar{r}$.
The cost to produce power by a generator is modeled as a quadratic function 
$f_{ \gen,i }^\tau = \alpha_1 p_{\gen ,i}^\tau + \alpha_2 (p_{\gen ,i}^\tau)^2$ for 
some $\alpha_1 > 0$  and $\alpha_2> 0$.
Storage devices are collected in $\STOR$ and their power 
$p_{\stor,i}^{\tau}$ satisfies bounds and a dynamical constraint given by
$-d_{\stor} \le p_{\stor,i}^{\tau} \le c_{\stor}$, $\tau \in [0,T]$,
$ q_{\stor,i}^{\tau+1} = q_{\stor,i}^{\tau} + p_{\stor ,i}^{\tau}$,
$\tau \in [0,T-1]$, and $0 \le q_{\stor,i}^{\tau} \le q_{\text{max}}$,
$\tau \in [0,T]$, 
where the initial capacity $q_{\stor ,i}^{0}$ is given and 
$d_{\stor }$, $c_{\stor }$ and $q_{\text{max}}$ are positive scalars.
There are no costs associated with the stored power.
Controllable loads are collected in $\CLOAD$ and their power is denoted by
$p_{\cload,i}^{\tau}$. A desired load profile $p_{\des,i}^\tau$ for $p_{\cload,i}^\tau$ is given and 
the controllable load incurs in a cost $f_{ \cload,i }^\tau =\beta \max \{0, p_{\des,i}^\tau  - p_{\cload,i}^\tau \}$, $\beta \ge0$,
if the desired load is not matched.
Finally, the device $i=N$ is the connection node with the main grid; its power
$p_{\trade}^{\tau}$ must satisfy
$| p_{\trade }^{\tau} | \le E$, $\tau \in [0,T]$.
The power trading cost is modeled as $f_{\trade}^\tau = -c_1 p_{\trade }^{\tau} + c_2 |p_{\trade }^{\tau} | $, 
with $c_1>0$ and $c_2>0$ being respectively the price and
the transaction cost.

The power network must satisfy a given power demand $D^\tau$ modeled by a
coupling constraint among the units
\begin{align*}
  \sum_{i\in \GEN} p_{\gen,i}^{\tau} +
  \!\!\!
  \sum_{i\in \STOR} p_{\stor,i}^{\tau} +
  \!\!\!
  \sum_{i\in \CLOAD} p_{\cload,i}^{\tau}
  +
  p_{\trade}^{\tau}- D^\tau
  \!\!= 0,
\end{align*}
for all $\tau\in[0,T]$. Reasonably, we assume $D^\tau$ to be known only by the
connection node $\trade$.

\subsection{Numerical Results}
\label{subsec:numerical_results}
We consider a heterogeneous network of $N = 10$ units with
$4$ generators, $3$ storage devices, $2$ controllable loads 
and $1$ connection to the main grid.
We assume that in the distributed MPC scheme each unit predicts its power 
generation strategy over a horizon of $T = 12$ slots.
In order to fit the microgrid control problem in our set-up, we let each $\sx{i}$ be 
the whole trajectory over the prediction horizon $[0,T]$, e.g.,
$  \sx{i}  = 
  \begin{bmatrix}
    p_{\gen ,i}^{0},\ldots, p_{\gen ,i}^{T}
  \end{bmatrix}{}^{\!\!\top}$,
for all the generators $i \in \GEN$ and, consistently, for the other device types.
As for the cost functions we define $f_i(\sx{i}) = \sum_{\tau=0}^T f_{\gen, i}^\tau (p_{\gen,i}^\tau)$
for $i \in \GEN$ and, similarly, for the other device types. 
The local (rate) bounds and the dynamics give rise to 
local constraints $X_i$.

In Figure~\ref{fig:constraint_violation} we show the algorithmic evolution of
the sum of the penalty parameters $\rho_i^t$ and the maximum violation of the
coupling constraint at each iteration $t$. 
As claimed in Theorem~\ref{thm:convergence},
$\sum_{i=1}^N \rho_i^t$ asymptotically goes to zero.
In this particular instance we also notice that, after the very first iterations, the
generated points are strictly feasible for the coupling constraints and hit 
the boundary in the limit from the interior.
We point out that feasibility of the coupling constraints is obtained during the transient, 
even if some $\rho_i^t$ are positive, so that the algorithm would not work
without the relaxation strategy.
In Figure~\ref{fig:tracking} we show how $\sum_{j\in\nbrs_i} ( \slambda{ij}^t - \slambda{ji}^t )$
compares with the unknown part of the coupling constraints of each agent $i$, 
namely $\sum_{j\neq i} \bg_j  (\bx_j^t)$.
The picture highlights that $\sum_{j\in\nbrs_i} ( \slambda{ij}^t - \slambda{ji}^t )$ 
actually distributedly ``tracks''  the maximum of the contribution in the coupling constraint 
due to all the other agents $j\neq i$ in the network. 

\begin{figure}[ht]\centering
  \begin{minipage}[b]{.48\linewidth}
  \centering
    \includegraphics[scale=1]{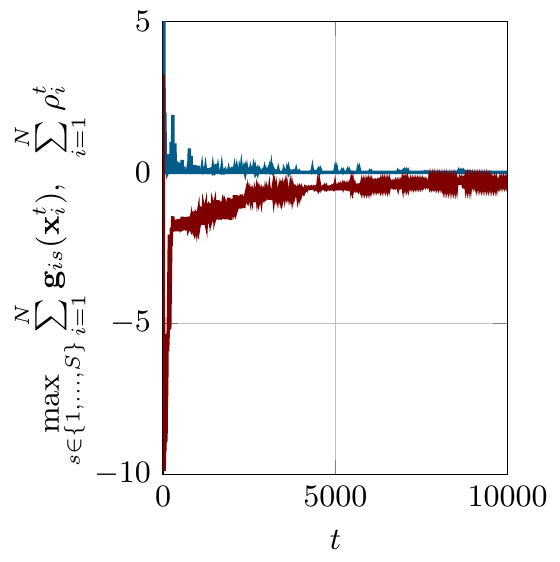}
  \caption{
      Evolution of the maximum violation of coupling constraints
      (red)
    and the sum of local violations %
    (blue).
}
  \label{fig:constraint_violation}
\end{minipage}
\hfill
\begin{minipage}[b]{.48\linewidth}
  \centering
  \includegraphics[scale=1]{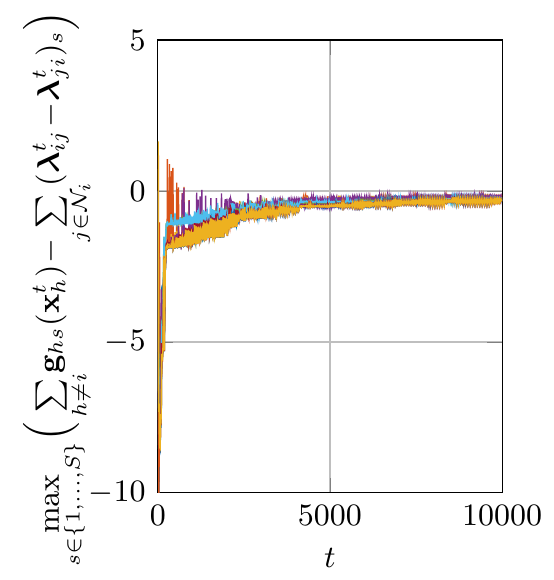}
  \caption{
    Evolution of the ``constraint tracking'' error of the coupling constraint for all $i\in\until{N}$.
  }
  \label{fig:tracking}
\end{minipage}
\end{figure}\vspace*{-1ex}

Finally, in Figure~\ref{fig:convergence} it is shown the convergence rate of the
distributed algorithm, i.e., the difference between the centralized optimal
cost $\eta^\star = f^\star$ and the sum of the local costs $\sum_{i=1}^N ( f_i(\bx_{i}^t) + M\rho_i^t)$,
normalized by $|f^\star|$. It can be seen that the proposed algorithm converges 
in a nonmonotone fashion to the optimal cost with a sublinear rate, as expected for 
a subgradient method.\vspace*{-1.5ex}
\begin{figure}[htbp]
\centering
  \includegraphics[scale=1]{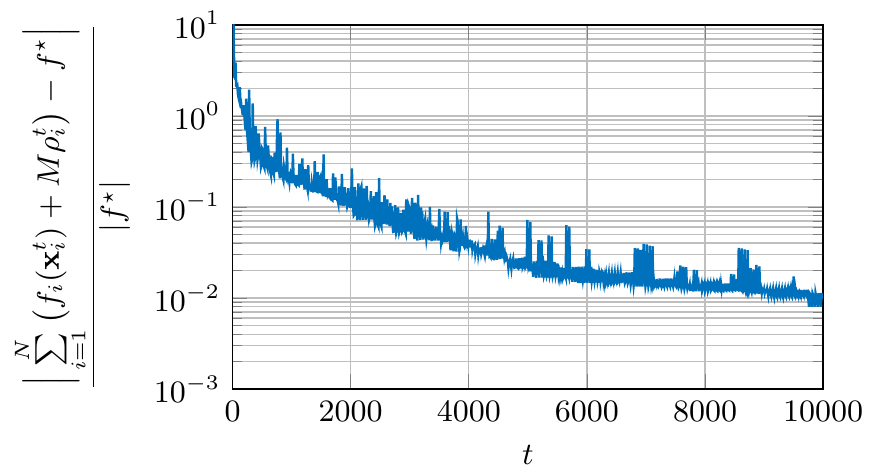}
  \caption{Evolution of the cost error
    $| \sum_{i=1}^N \big( f_i( \bx_{i}^t) + M\rho_i^t \big) - f^\star|/|f^\star|$.
    }\vspace*{-2.5ex}
\label{fig:convergence}
\end{figure}

\section{Conclusions}
\label{sec:conclusions}
In this paper we have proposed a novel distributed method to solve
constraint-coupled convex optimization problems in which a separable cost
function is minimized subject to both local constraints involving one component of
the decision vector and coupling constraints involving all the components. While the
algorithm has a very simple structure (a local optimization and a linear
update), its analysis involves a relaxation approach and a deep tour into
duality theory showing both the convergence to the optimal cost and the 
primal recovery. In particular, this last property allows each node to compute its
portion of the optimal solution without resorting to any averaging
mechanism, which is instead commonly required in methods based on 
dual decomposition. 
Numerical computations on an instance of a cooperative Distributed Model Predictive
Control scheme in smart microgrids have corroborated the theoretical results.
\vspace*{-1.5ex}

\section*{Acknowledgment}
The authors would like to thank Andrea Camisa for the deep discussions and 
useful suggestions.
\vspace*{-1.0ex}

\begin{small}
  \bibliographystyle{IEEEtran}
  \bibliography{distributed_dual_dual}
\end{small}

\end{document}